\documentclass[letterpaper,12pt]{article}
\usepackage{tabularx} 
\usepackage{amsmath}  
\usepackage{graphicx} 
\usepackage[margin=1in,letterpaper]{geometry} 
\usepackage{cite} 
\usepackage[final]{hyperref} 
\hypersetup{
	colorlinks=true,       
	linkcolor=blue,        
	citecolor=blue,        
	filecolor=magenta,     
	urlcolor=blue         
}
\usepackage{blindtext}
\usepackage[utf8]{inputenc} 
\usepackage[T1]{fontenc}    
\usepackage{url}            
\usepackage{booktabs}       
\usepackage{amsfonts}       
\usepackage{nicefrac}       
\usepackage{microtype}      
\usepackage{lipsum}
\usepackage[braket]{qcircuit}
\usepackage{subfigure}
\usepackage{amsthm}
\usepackage{mathtools}
\usepackage[ruled,linesnumbered]{algorithm2e}
\usepackage{authblk}
\usepackage{float}
\usepackage{enumerate}
\usepackage{tablefootnote}

\newtheorem{theorem}{Theorem}
\newtheorem{lemma}[theorem]{Lemma}

\begin{document}

\title{Implementation of Continuous-Time Quantum Walk on Sparse Graph}
\author{Zhaoyang Chen\thanks{chenchy76@mail2.sysu.edu.cn}}
\author{Guanzhong Li\thanks{ligzh9@mail2.sysu.edu.cn}}
\author{Lvzhou Li\thanks{lilvzh@mail.sysu.edu.cn(corresponding author)}}
\affil{Institute of Quantum Computing and Software, \\School of Computer Science and Engineering, \\Sun Yat-sen University, Guangzhou 510006, China}
\date{\today}
\maketitle

\begin{abstract} Continuous-time quantum walks (CTQWs) play a crucial role in quantum computing, especially for designing  quantum algorithms. However, how to efficiently implement CTQWs is a challenging issue. In this paper, 
we study  implementation of CTQWs on sparse graphs, i.e.,  constructing efficient quantum circuits for implementing the unitary operator $e^{-iHt}$, where $H=\gamma A$ ($\gamma$ is a constant and $A$ corresponds to the adjacency matrix of a graph). Our result is, for a $d$-sparse graph with $N$ vertices and evolution time $t$, we can approximate $e^{-iHt}$ by a quantum circuit with  gate complexity $(d^3 \|H\| t N \log N)^{1+o(1)}$, compared to the general Pauli decomposition, which scales like $(\|H\| t N^4 \log N)^{1+o(1)}$. For sparse graphs, for instance, $d=O(1)$, we obtain a noticeable improvement. Interestingly, our  technique is related to graph decomposition. More specifically, we decompose the graph into a union of star graphs, and correspondingly, the Hamiltonian $H$ can be represented as the sum of some Hamiltonians $H_j$, where each $e^{-iH_jt}$ is a CTQW on a star graph which can be implemented efficiently.
\end{abstract}


\section{Introduction}

Due to the potential applications in a variety of problems, e.g., spatial searching\cite{childs2004spatial}, testing graph isomorphism\cite{douglas2008classical} and simulating quantum dynamics\cite{mohseni2008environment}, quantum walks have recently attracted significant attention both theoretically and experimentally\cite{farhi1998quantum,childs2009universal,childs2013universal,wang2013physical,li2024optimal,li2024recovering}. To execute quantum walk-based algorithms on quantum computers, the development of efficient quantum circuits for implementing quantum walks is essential.

There are two distinct kinds of quantum walks: the discrete-time quantum walks (DTQWs)\cite{aharonov1993quantum,hillery2003quantum,szegedy2004quantum} and the continuous-time quantum walks (CTQWs)\cite{farhi1998quantum}. While DTQWs have attained notable advancements in circuit design\cite{douglas2009efficient,loke2012efficient,loke2017efficient}, the implementation of CTQWs seems to be more challenging\cite{loke2017efficient1}. On the one hand, CTQWs demand a sophisticated time-dependent quantum circuit implementation, as the evolution operator $U(t)$ is intricately defined as a continuous function of $t$. In contrast, DTQWs operate with discrete time steps where the evolution operator $U(t)$ consists of the repetition of a single time-step $U$, which remains independent of $t$. On the other hand, the single time-step operator $U$ in DTQWs acts solely on the local structure of the graph, propagating amplitudes only to adjacent vertices, whereas the dynamics of CTQWs are characterized by the propagation of amplitudes between any two connected vertices on the graph, necessitating a consideration of the global graph structure in quantum circuit design.

Efficient quantum circuit implementation for CTQWs has been demonstrated across various graph classes, such as commuting graphs and Cartesian products of graphs\cite{loke2017efficient1}. Notably, on certain graph classes, including  complete bipartite graphs, complete graphs\cite{loke2017efficient1}, glued trees\cite{childs2003exponential} and star graphs\cite{qu2022deterministic}, CTQWs can be efficiently implemented using diagonalization techniques. Additionally, a specific set of circulant graphs can be efficiently implemented, leveraging the quantum Fourier transform\cite{qiang2016efficient}. For circulant graphs with $2^n$ vertices possessing $O(\text{poly}(n))$ non-zero eigenvalues or efficiently characterizable distinct eigenvalues (e.g., cycle graphs, Möbius ladder graphs), an efficient quantum circuit for CTQWs can be constructed\cite{qiang2016efficient}.

The problem of implementing CTQWs can also be regarded as Hamiltonian simulation, which is BQP-complete\cite{feynman1985quantum}, suggesting that classical algorithms solving this problem efficiently are improbable. Lloyd\cite{lloyd1996universal} pioneered the first explicit quantum algorithm for Hamiltonian simulation, focusing on local interactions, which was later extended to more general sparse Hamiltonians by Aharonov and TaShma\cite{aharonov2003adiabatic}. Subsequent work has obtained celebrated achievements over the years\cite{berry2007efficient,berry2009black,berry2014exponential,berry2015simulating,childs2011simulating,childs2012hamiltonian,low2019hamiltonian,low2017optimal}, but mostly in the "black-box" setting except for those algorithms based on product formulas, which is also most commonly utilized. Product formulas rely on the assumption that the Hamiltonian $H$ can be broken down into a sum of individual terms ($H=\sum_{j=1}^m H_j$), such that each exponential $e^{-iH_jt}$ can be implemented efficiently. This allows the approximation of the time-evolution operator $e^{-iHt }$ by the products of exponentials $e^{-iH_j t}$. The size of the quantum circuit depends on the number of terms $m$. Thus, the way we decompose Hamiltonian is directly related to the gate complexity. Typically, any $2^n\times 2^n$-dimensional $H$ can be decomposed as a linear combination of Pauli string ($H = \sum_{j=1}^m h_j P_j$, $P_j \in \{I,X,Y,Z\}^{\otimes n}$), where each $e^{-ih_j P_jt}$ can be implemented by $O(n)$ elementary gates, but leading to the number of terms $m=O(4^n)$. In this paper, we adopt a graph decomposition approach to decompose a $d$-sparse graph into $m=6d$ star forests, and achieve a fast-forwarded quantum circuit implementation of CTQWs on star graphs. Consequently, this method results in a more efficient circuit implementation for CTQWs on $d$-sparse graphs compared to the Pauli decomposition method.


\section{Problem Description and Main Result}\label{sec:problem}

Let us consider an undirected graph $G(V,E)$ comprising a vertex set $V=\{v_1,v_2,...,v_N\}$ $(\text{where }N=2^n)$ and an edge set $E=\{v_iv_j,...,v_kv_l\}$. The complete depiction of the graph $G$ can be encapsulated by the $N\times N$ adjacency matrix $A$, which is defined as:
\begin{equation}
    A_{ij}=\left\{\begin{matrix}
     1 & v_iv_j \in E,\\
     0 & \text{otherwise.}
    \end{matrix}\right.    
\end{equation}
In the case of an undirected graph, the corresponding adjacency matrix $A$ is symmetric ($A=A^T$), and the degree of a vertex $v_i$ is represented by the number of edges connected to it, denoted as:
\begin{equation}
    d_i = \sum_{j=1}^N A_{ij}.
\end{equation}
We can further define an $N\times N$ degree matrix $D$ by $D_{ij}=d_i\delta_{ij}$, where $\delta_{ij}$ equals $1$ if $i=j$, and $0$ otherwise. Assuming that the maximum degree of graph $G$ is $d$, i.e., for all $i=1,...,N$, $d_i \le d$, the adjacency matrix $A$ is consequently $d$-sparse\footnote{A $d$-sparse matrix has at most $d$ non-zero entries in each row or each column.}.

\paragraph{Continuous-time quantum walk model} The CTQW model originates from Markov chain, assuming that there is a quantum particle walking on graph $G$, and a quantum state $|\psi(t)\rangle$ describes the particle's position at time $t$, with $p_j(t) = |\langle \psi(t)|j\rangle|^2$ denoting the probability that this particle found at vertex $j$, where the computational basis $\{|1\rangle, |2\rangle, \dots, |N\rangle\}$ corresponds to the vertices in graph $G$. The evolution of the state $|\psi(t)\rangle$ is governed by the Schr$\mathrm{\ddot o}$dinger equation as follows,
\begin{equation}\label{eq:schr}
    i\frac{\mathrm{d} }{\mathrm{d} t} |\psi(t)\rangle = H|\psi(t)\rangle,
\end{equation}
where the choice of Hamiltonian $H$ depends on the purpose, for example, $H=\gamma(D-A)$\cite{childs2002example} and $H=\gamma A$\cite{chakrabarti2012design,farhi1998quantum} for graph isomorphism problems ($\gamma$ is the probability per unit time of moving to an adjacent vertex), and $H=|\omega \rangle \langle \omega| - A$\cite{qiang2021implementing} for searching the marked vertex $|\omega\rangle$. In this paper we only focus on the choice of $H=\gamma A$.

The solution to equation (\ref{eq:schr}) is $|\psi(t)\rangle = e^{-iHt}|\psi(0)\rangle$, and our goal is to construct a quantum circuit $\tilde{U}$, which implements the target operator $e^{-iHt}$ within error $\epsilon$:
\begin{equation} \label{eq:error}
    \|\tilde{U}-e^{-iHt}\| < \epsilon.
\end{equation}

Product formulas method, introduced by Lloyd\cite{lloyd1996universal} for the problem of Hamiltonian simulation, is commonly used to simulate the evolution $e^{-iHt}$ when Hamiltonian has the form $H=\sum_{j=1}^m H_j$, where each $H_j$ can be efficiently simulated for arbitrary evolution time $t'$. Then the evolution $e^{-iHt}$ can be approximated by a product of exponentials $e^{-iH_jt'}$. The upper bound of the number of exponentials required is given in reference\cite{berry2007efficient}, as shown in Lemma \ref{theorem:pf}. 

\begin{lemma}[\cite{berry2007efficient}, Theorem 1]\label{theorem:pf}
    In order to simulate the Hamiltonian of the form $H=\sum_{j=1}^m H_j$ for time $t$ by a product of exponentials $e^{-iH_jt'}$ within error $\epsilon$, the number of exponentials $N_{\exp}$ is bounded by
    \begin{equation}
        N_{\exp} \le 2m^25^{2k} \|H\|t(m\|H\|t/\epsilon)^{1/2k},
    \end{equation}
    for $\epsilon \le 1 \le 2m5^{k-1} \|H\| t$, where $k$ is an arbitrary positive integer.
\end{lemma}

The total gate count of quantum circuit $\tilde{U}$ is upper bounded by $N_{\exp} \times \max_j C(H_j)$, where $C(H_j)$ is the cost required to implement exponential $e^{-iH_jt'}$.

In this paper our aim is to construction a quantum circuit implementing the CTQW $e^{-iHt}$ on a graph $G$ where $H=\gamma A$, $A$ is the adjacency matrix of $G$, and $\gamma$ is a constant.
Our technique is to decompose the graph $G$ into a union of star graphs $G_j$, and correspondingly, the Hamiltonian $H$ can be represented as the sum of some Hamiltonians $H_j$, where each $e^{-iH_jt}$ is a CTQW on a star graph $G_j$ which can be implemented efficiently. The idea  is illustrated in Fig.\ref{fig:dec_example}. The graph $G$ is decomposed into a union of star graphs $G=G_1\cup G_2\cup G_3$, and then the corresponding adjacency matrices satisfy $A=A_1+A_2+A_3$. Therefore, we obtain a decomposition of the Hamiltonian $H=\sum_{j=1}^3 H_j$, and each exponential $e^{-iH_jt}$ can be implemented efficiently. Then, by Lemma \ref{theorem:pf} we can construct a quantum circuit that implements $e^{-iHt}$. The main result of this work is shown in Theorem \ref{theorem:result}. The details will be demonstrated in the next section.

    \begin{figure}[htb]
    \centering
    \includegraphics[width=\linewidth]{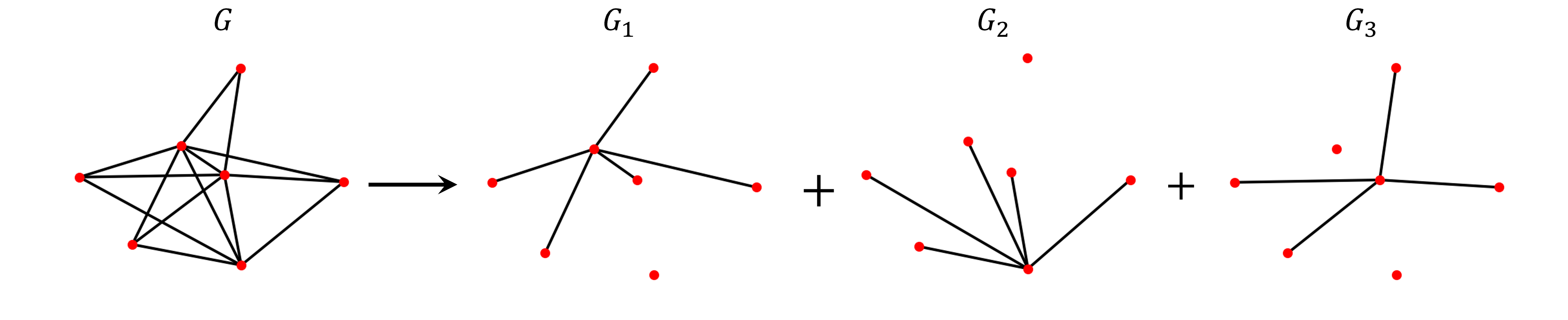}
    \caption{A simple example of the graph decomposition.}
    \label{fig:dec_example}
    \end{figure}

\begin{theorem}\label{theorem:result}
 Let $A$ be the adjacency matrix of a $d$-sparse graph $G$ with $N$ vertices. There exists a quantum circuit that implements the time-evolution operator $e^{-iHt}$ for CTQW on graph $G$ with error at most $\epsilon$ in spectral norm, using $O(d^{3+1/2k}5^{2k}(\|H\|t)^{1+1/2k} N \log N /\epsilon^{1/2k})$ elementary gates, where $H=\gamma A$, $\epsilon \le 1 \le 12d5^{k-1}\|H\|t$ and $k$ is an arbitrary positive integer.
    
\end{theorem}


\section{Quantum Walk Implementation}\label{sec:decom}

In this section we describe our algorithm for implementing CTQWs on sparse graphs, as shown in Algorithm \ref{alg:CTQW implementation}. The method depends on  efficient implementation of CTQWs on star graphs, which was studied in \cite{qu2022deterministic}, but the circuit constructed in \cite{qu2022deterministic} seems to be inappropriate for our case, because the star graph in \cite{qu2022deterministic}, represented by $S = \sum_{x=1}^{N} (|0\rangle \langle x| + |x\rangle \langle 0|)$ is connected, and the central vertex is encoded as $|0\rangle$ while other vertices are encoded as $\{|1\rangle, |2\rangle, \dots, |N\rangle \}$. In our situation, the star graph $S = \sum_{x\in V'} (|c\rangle \langle x| + |x\rangle \langle c|)$ is not connected (there exist some isolated vertices in $S$), and the center $|c\rangle$ can be encoded as any integer in $[N]$ while other vertices in $V'$ are a subset of the vertex set. Thus we construct a different circuit for CTQWs on star graphs as shown in Fig.\ref{fig:star ctqw}.

\begin{algorithm}

\caption{Implementation of CTQWs on sparse graphs}
\label{alg:CTQW implementation}
\begin{enumerate} [Step 1]
    \item Decompose a $d$-sparse graph $G$ into $d$ forests $\{F_i\}^d_{i=1}$ of rooted trees by Lemma \ref{lemma:forest decomposition}.
    \item Color the vertices of each forest $F_i$ by Lemma \ref{lemma:coloring} using $6$ colors, then divide \\ each edge set $E_i$ into $6$ subsets $E_i = \bigcup_{k=1}^6 E_{ik}$ according to the color of vertices \\that the directed edges point towards.
    \item Separately construct quantum circuit of CTQW on each graph $F_{ik}\coloneqq (V,E_{ik})$.
    \item Obtain the implementation of CTQW on graph $G$ by Lemma \ref{theorem:pf}.
\end{enumerate}

\end{algorithm}  

\paragraph{Step 1}
We introduce Lemma \ref{lemma:forest decomposition} here for decomposing any $d$-sparse graph $G$, which was implied in \cite{panconesi2001some}. For completeness, we give its proof in Appendix \ref{sec:proof of forest dec}.
\begin{lemma} [Forest decomposition] \label{lemma:forest decomposition}
    For any $d$-sparse graph $G(V,E)$, there exists a $d$-partition of the edge set $E=\bigcup_{i=1}^d E_i$, suc that each $F_i\coloneqq (V,E_i)$ is a forest. Furthermore, there exists a direction assignment of the edges such that all $F_i$ consist of rooted trees\footnote{A rooted tree is a directed tree with a unique root.}.
\end{lemma}

\paragraph{Step 2} 
We call a parallel graph coloring algorithm to produce a valid $6$-coloring\footnote{A valid vertex coloring is an assignment of colors to each vertex such that each vertex has a different color from the vertices adjacent to it.} of a rooted tree in $O(\log^* N)\footnote{If $N > 1$, $\log^* N = 1 + \log^* \log N$, else $\log^* N=0$.}$ time in Lemma \ref{lemma:coloring}. In fact, any tree can be $2$-colored according to the depth of each vertex is odd or even, but it is hard to achieve this in parallel, which results in more time cost.

\begin{lemma} \label{lemma:coloring}
    There exists a valid vertex coloring with $6$ colors for any rooted tree.
\end{lemma}

\begin{proof}

The vertex coloring is based on an extension of "deterministic coin-flipping" technique\cite{goldberg1987parallel}. Denote $C_v$ as the color of vertex $v$ and $p_v$ be the parent of $v$. Starting from a trivial valid vertex coloring using the label of vertices (i.e., $C_v=v$), then we iteratively reduce the number of colors by recoloring each vertex $v$. In detail, denote $k$ as the index of the first bit such that $C_v[k] \ne C_{p_v}[k]$. Then we recolor vertex $v$ with the concatenation of $k$ and $C_v[k]$, denoted as $\left \langle k,C_v[k]\right \rangle$. If the vertex $v$ is a root, take $k=0$ and $C_v[0]$. 

We can prove that coloring computed by this procedure is valid. Recall that $p_v$ is the parent of vertex $v$ and let $u$ be one of the children of $v$. Denote $i$ as the index of the first bit that $C_{p_v}[i] \ne C_v[i]$ and $j$ as the index of the first bit of $C_v[j] \ne C_u[j]$. Then the color of $v$ is updated to $\left \langle i,C_v[i]\right \rangle$ and the color of $u$ is $\left \langle j,C_u[j]\right \rangle$. Assuming that $i\ne j$, the new color of $v$ and $u$ are obviously different. Otherwise, if $i=j$, $C_v[j]$ is not equal to $C_u[j]$ according to the definition of $j$, so the new color of $v$ and $u$ are different as well. Because each vertex is colored differently from its children, the validity is preserved. 

The number of colors will be reduced to $6$ after $O(\log^* N)$ iterations. Denote $L_j$ as the number of bits used to represent the color after $j$ iterations. Initially, $L_0=\lceil \log N \rceil$, and we can derive that $L_{j+1} = \lceil \log L_j \rceil +1$. After $j=O(\log^* N)$ iterations, $L_j$ decrease to $3$\cite{goldberg1987parallel}, and this number can not be further decreased, for $L_{j+1} = L_j = 3$. Hence the number of colors is at most $6$, by $3$ possible values for $k$ and $2$ possible values for $C_v[k]$. 
    
\end{proof}

After coloring the vertices in $F_i$, we divide the edge set $E_i= \bigcup_{k=1}^6 E_{ik}$ according to the color of vertices that edges point to. Note that each connected component in graph $F_{ik}\coloneqq (V,E_{ik})$ is a star.

\paragraph{Step 3}

We have shown that each graph $F_{ik}$ is actually a disjoint union of star graphs, and the corresponding Hamiltonian $H_{ik} = \gamma \sum_{l=1}^{L_{ik}} S_{ikl}$, where each $S_{ikl}$ is the adjacency matrix of a star graph in $F_{ik}$, and $L_{ik}$ is the number of stars. Due to the commutativity between all $S_{ikl}$ in $H_{ik}$, we have
\begin{equation} \label{eq:commutativity}
    e^{-iH_{ik}t}=\prod_{l=1}^{L_{ik}} e^{-i\gamma S_{ikl} t}.
\end{equation}
Suppose $S_{ikl} = \sum_{x\in V'}(|c\rangle \langle x|+ |x\rangle \langle c|)$, where $c$ is the central vertex and $V'$ is the vertex set of $S_{ikl}$ except for the center. The non-zero eigenvalues of this matrix is $\lambda_{\pm} = \pm \sqrt{|V'|}$, where $|V'|$ denote the cardinal number of $V'$, and the eigenvectors are
\begin{equation} \label{eq:star eigvec}
    |v_{\pm} \rangle = \frac{1}{\sqrt{2}} \left(|c\rangle \pm \frac{1}{\sqrt{|V'|}} \sum_{x\in V'} |x\rangle \right).
\end{equation}

In order to implement $e^{-i\gamma S_{ikl} t}$ as shown by Theorem \ref{theorem:ctqw on star} below, we first introduce Lemma \ref{lemma:sparse qsp}.
\begin{lemma} [Sparse quantum state preparation\cite{gleinig2021efficient}] \label{lemma:sparse qsp}
    Let $V\subset \{0,1\}^n$ be a subset of basis states, there exists a quantum circuit $C$ that maps the initial state $|0^n\rangle$ to the quantum state $$|\phi\rangle = \sum_{x \in V} c_x |x\rangle$$
    using $O(|V|n)$ elementary gates, where $c_x$ is the non-zero amplitude.
\end{lemma}
Note that recently more efficient approaches have been developed for sparse quantum state preparation in\cite{luo2024circuitcomplexity}, but Lemma \ref{lemma:sparse qsp} is sufficient for our purpose.

\begin{theorem} \label{theorem:ctqw on star}
    Suppose $S = \sum_{x\in V'} (|c\rangle \langle x| + |x\rangle \langle c|)$ is the adjacency matrix of a graph with $N$ vertices, where $V'$ is a subset of the vertex set. There exists a quantum circuit that implements $e^{-i\gamma St}$ with $O(|V'|\log N)$ gate complexity and one ancilla qubit.
\end{theorem}

\begin{proof}
    Assume we have an $n$-qubit working register and one ancilla qubit. Let quantum circuit $C_{1} \in \{I,X\}^{\otimes n}$ performs $|0\rangle \mapsto |c\rangle,$
    and circuit $C_2$ prepares the quantum state $|\phi\rangle = \frac{1}{\sqrt{|V'|}} \sum_{x\in V'} |x\rangle$ by Lemma \ref{lemma:sparse qsp}.
    Then the circuit in Fig. \ref{fig:star qsp} performs as follows:
    \begin{equation} \label{eq:star map}
        |0\rangle |0\rangle_a \mapsto |v_+\rangle |0\rangle_a,\ |1\rangle |0\rangle_a \mapsto |v_-\rangle |0\rangle_a.
    \end{equation}

    We first apply a Hadamard gate and a swap gate to produce  $$\frac{1}{\sqrt{2}}(|0\rangle |0\rangle_a \pm |0\rangle |1\rangle_a),$$
    where $\pm$ corresponds to $|0\rangle |0\rangle_a$ and $|1\rangle|0\rangle_a$ respectively, and then we perform $C_1$ or $C_2$ controlled by the value of the ancilla qubit to get $$\frac{1}{\sqrt{2}} (|c\rangle |0\rangle_a \pm |\phi \rangle |1\rangle_a).$$
    Next we apply a Pauli $X$ gate on the ancilla qubit if the value of the working register equals to $c$, and the state is $$\frac{1}{\sqrt{2}}(|c\rangle \pm |\phi\rangle )|1\rangle_a.$$
    Finally we apply a Pauli $X$ gate on the ancilla qubit to recover it, then we obtain the map as (\ref{eq:star map}).

    \begin{figure}[htb]
    \centerline{
    \Qcircuit @C=1em @R=.7em {
     &\qw   & \qw    &\qw &  \multigate{4}{C_1} & \qw &  \multigate{4}{C_2} & \qw & \gate{} & \qw  & \qw \\
     & \qw   & \qw    &\qw &  \ghost{C_1} & \qw     & \ghost{C_2} & \qw&\gate{} \qwx &\qw &\qw \\
     & \vdots &   &   &       &  &  &  & \qwx & & \\
     & \qw   & \qw    & \qw & \ghost{C_1} & \qw     & \ghost{C_2} & \qw & \gate{}  &\qw & \qw \\
     & \gate{H}  & \qswap & \qw & \ghost{C_1} & \qw     &  \ghost{C_2} & \qw & 
     \gate{} \qwx  &\qw  &\qw \\
     & \qw   & \qswap \qwx & \qw  & \ctrlo{-1}& \qw   & \ctrl{-1}   & \qw &\gate{X} \qwx  & \gate{X} & \qw &  \\
    }
    }
    \caption{Quantum circuit for transforming from eigenbasis of $S$ to the computational basis with single ancilla qubit.}
    \label{fig:star qsp}
    \end{figure}
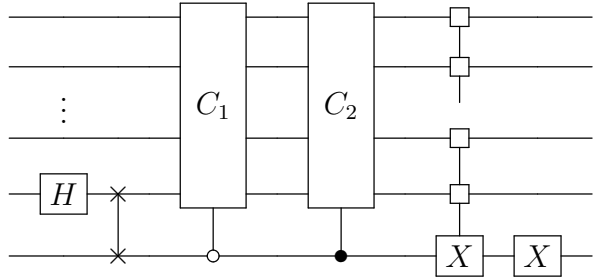

    There are only two non-zero eigenvalues, thus the diagonal operator can be easily implemented by applying $R_z(2\gamma t\sqrt{|V'|})$ to the least significant qubit in the working register, controlled by all other qubits in $|0\rangle$ state. We know that the $H$ and $R_z$ gates can be combined as $R_x$ gate. Hence the quantum circuit for CTQW on $S$ is as Fig. \ref{fig:star ctqw}. By Lemma \ref{lemma:sparse qsp}, the size of circuit $C_2$ is $O(|V'|\log N)$, and other gates in Fig. \ref{fig:star ctqw} cost $O(\log N)$ elementary gates, including the multi-controlled rotation gate\cite{nie2024quantum}. Thus the gate complexity of this circuit is $O(|V'|\log N)$.

    \begin{figure}[htb]
    \centerline{
    \Qcircuit @C=1em @R=.7em {
        & \qw       & \gate{}       & \multigate{4}{C_2^\dagger}    & \multigate{4}{C_1^\dagger}    & \qw           & \ctrlo{1}             & \qw           & \multigate{4}{C_1}    & \multigate{4}{C_2}    & \gate{} & \qw & \qw \\
        & \qw       & \gate{} \qwx  & \ghost{C_2^\dagger}           & \ghost{C_1^\dagger}           & \qw           & \ctrlo{1}             & \qw           & \ghost{C_1}           & \ghost{C_2}           & \gate{} \qwx & \qw & \qw \\
        & \vdots    & \qwx          &                               & 
                                      &               &                     & & & & \qwx  & & \\ 
        & \qw       & \gate{}   & \ghost{C_2^\dagger}           & \ghost{C_1^\dagger}           & \qw           & \ctrlo{1}             & \qw           & \ghost{C_1}           & \ghost{C_2}           & \gate{} 
 & \qw & \qw \\
        & \qw       & \gate{} \qwx  & \ghost{C_2^\dagger}           & \ghost{C_1^\dagger}           & \qswap        & \gate{R_x(\alpha)}    & \qswap        & \ghost{C_1}           & \ghost{C_2}           & \gate{} \qwx & \qw & \qw \\
        & \gate{X}  & \gate{X} \qwx & \ctrl{-1}                     & \ctrlo{-1}                    &\qswap \qwx    & \qw                   & \qswap \qwx   & \ctrlo{-1}            & \ctrl{-1}             & \gate{X} \qwx & \gate{X} & \qw \\
    }
    }
    \caption{Efficient quantum circuit for CTQW $e^{-i\gamma St}$ over $S = \sum_{x\in V'} (|c\rangle \langle x| + |x\rangle \langle c|)$ with single ancilla qubit,  where $\alpha = 2\gamma\sqrt{|V'|}t$.}
    \label{fig:star ctqw}
    \end{figure}
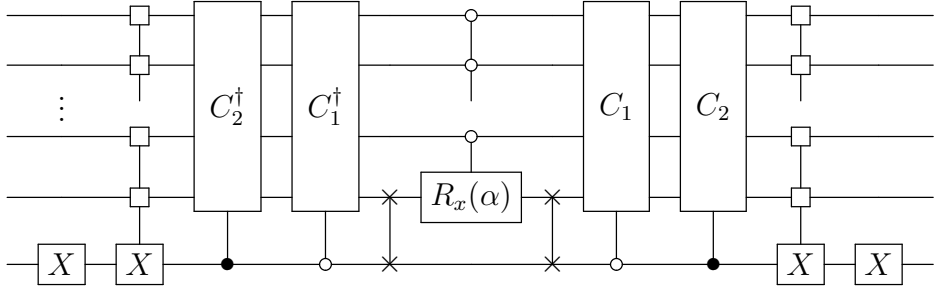

\end{proof}

By Theorem \ref{theorem:ctqw on star}, the gate cost of implementing each $e^{-i\gamma S_{ikl}t}$ is $O(d \log N)$ with $|V'| = O(d)$. Therefore, the gate complexity of implementing $e^{-iH_{ik}t}$ is $O(N d \log N)$ according to (\ref{eq:commutativity}), for $L_{ik}=O(N)$.

\paragraph{Step 4}

By Lemma \ref{lemma:forest decomposition} and Lemma \ref{lemma:coloring}, we obtain $H=\sum_{i=1}^d \sum_{k=1}^6 H_{ik} = \sum_{j=1}^{6d} H_j$, where each $e^{-iH_jt}$ can be implemented with $O(Nd \log N)$ cost by Step 3. And by Lemma \ref{theorem:pf} we can learn the upper bound of the required number of exponentials $N_{\exp}$. Thus, the gate cost required to implement $e^{-iHt}$ within error $\epsilon \le 1 \le 12d5^{k-1}\|H\|t$ is as follows:
\begin{equation}
\begin{aligned}
    C &= N_{\exp} \times \max_j{C(H_j)} \\
    &\le 2m^25^{2k} \|H\|t(m\|H\|t/\epsilon)^{1/2k} \times O(N d\log N) \\
    &= O(d^{3+1/2k}5^{2k}(\|H\|t)^{1+1/2k} N \log N /\epsilon^{1/2k}),
\end{aligned}
\end{equation}
where we have used the fact that $m=6d$. Therefore we achieve the result presented in Theorem \ref{theorem:result}, i.e., we can implement CTQW on a $d$-sparse graph by a quantum circuit with $ (d^3 \|H\| t N \log N)^{1+o(1)}$ elementary gates.

Generally, in order to implement $e^{-iHt}$, we will decompose $H$ as the linear combination of Pauli strings, i.e., $H=\sum_{j=1}^m h_j P_j$, where each $P_j \in \{I,X,Y,Z\}^{\otimes n}$ and $m=O(4^n)$. The gate complexity of implementing exponential $e^{-ih_j P_j t}$ is $O(n)$. Then by Lemma \ref{theorem:pf}, the gate cost for implementing $e^{-iHt}$ is $(\|H\| t N^4 \log N)^{1+o(1)}$. Thus our method is much better in the sparse case, e.g., $d=O(1)$. Note that for dense graphs (i.e., $d=O(N)$), the complexity is the same as the Pauli decomposition method.


\section{Discussion}

As shown in Section \ref{sec:decom}, we can implement CTQW on a sparse graph by decomposing the graph into forests consisting of stars, on which CTQWs can be efficiently implemented. Then according to Lemma \ref{theorem:pf} we can combine the Hamiltonians corresponding to these decomposed graphs to approximate the CTQW on the original graph. There may exist other graph decomposition techniques that can be used to accomplish this task. For instance, we can decompose the graph as an union of stars by finding a vertex cover of the graph. In this paper we use star graphs as the constituent graphs for our construction. There are a few other classes of graphs on which CTQWs can be implemented with efficient quantum circuits, such as complete graph, complete bipartite graph\cite{loke2017efficient1}, circulant graph with efficiently characterized eigenvalue spectrum\cite{qiang2016efficient} and the Cartesian product of these graphs\cite{loke2017efficient1}. Perhaps we can choose some of them as the constituent graphs for circuit construction. In addition, the method in this paper doesn't work well for CTQWs over dense graphs, and we wonder if there are other methods that can perform better than Pauli decomposition on general graphs. We are also curious about the efficient implementation of CTQWs with other choice of Hamiltonian like $H=\gamma(D-A)$ and $H=|\omega \rangle \langle \omega| - A$, as mentioned in Section \ref{sec:problem}.


\bibliographystyle{unsrt}
\bibliography{reference}

\appendix
\section{Proof of Lemma \ref{lemma:forest decomposition}} \label{sec:proof of forest dec}

\begin{proof}
Let $d_v$ denotes the degree of vertex $v$. Each vertex $v$ proposes a color for each edge incident on it by ranking these edges arbitrarily, then each edge receives a number between 1 and $d_v$. Therefore, each edge in the graph gets two proposals from its two endpoints. We choose the proposal of the vertex with higher label to be the color of the edge (i.e., if $u<v$ then edge $uv$ is colored with $v$'s proposal). For $i=1,...,d$, let $E_i$ denotes the edge set with color $i$. Then $F_i\coloneqq (V,E_i)$ is a forest subgraph of $G$ because there is not cycle inside any $F_i$. Note that any cycle in $F_i$ must contain a vertex $v$ has greater label than its two neighbors and the color of edges incident on $v$ in this cycle are both decided by $v$, which contradicts the fact that vertices propose different colors for different edges. 

Now we assign directions to the edges in forests, which orient from the endnode of smaller label to the one of higher label, then all forests consist of rooted trees. Observe that there is only one parent for each vertex in these forests, otherwise the vertex with more than one parent must propose the same color for different edges, which is not possible according to the procedure of coloring. Then there is a unique root for each directed tree in these forests, meaning that each tree is a rooted tree .
\end{proof}

\end{document}